\documentclass[a4paper,11pt]{amsart}
\usepackage{amsfonts,amssymb,epsfig,latexsym}
\usepackage{amsmath}
\usepackage[latin1]{inputenc}
\usepackage{color,layout}
\usepackage{ifthen}

\usepackage{dsfont}

\usepackage{pdfsync}
\usepackage{graphicx}
\usepackage{color}
\usepackage{pdfsync} 
\date{}

\newcommand{\be}{\begin{equation}}
\newcommand{\ee}{\end{equation}}

\def\R{\mathbb{R}}
\def\C{\mathbb{C}}

\def\O{\mathcal{O}}

\renewcommand{\Re}{\text{{\rm Re}\;}}
\renewcommand{\Im}{\text{{\rm Im}\;}}

\newtheorem{theorem}{Theorem}[section]
\newtheorem{lemma}[theorem]{Lemma}
\newtheorem{proposition}[theorem]{Proposition}

\theoremstyle{definition}

\newtheorem{remark}[theorem]{Remark}

\numberwithin{equation}{section}

\title[Widths of highly excited shape resonances]{Widths of highly excited shape resonances}
\author[Marzia Dalla Venezia \&
Andr\'e MARTINEZ]{
Marzia Dalla Venezia${}^1$ \&
Andr\'e MARTINEZ$ {}^2$
 }

\begin{document}

\maketitle
\addtocounter{footnote}{1}
\footnotetext{{\tt\small  Universit\`a di Bologna,
Dipartimento di Matematica, Piazza di Porta San Donato, 40127
Bologna, Italy,
marzia.dallavenezia3@unibo.it } }
\addtocounter{footnote}{1}
\footnotetext{{\tt\small Universit\`a di Bologna,
Dipartimento di Matematica, Piazza di Porta San Donato, 40127
Bologna, Italy,
andre.martinez@unibo.it }}
\begin{abstract}
We study the widths of shape resonances for the semiclassical multi-dimensional Schr\"odinger operator, in the case where the frequency remains close to some value strictly larger than the bottom of the well. Under a condition on the behavior of the resonant state inside the well, we obtain an optimal lower bound for the widths.
\end{abstract}
\vskip 4cm
{\it Keywords:} Resonances; life-time; semiclassical.
\vskip 0.5cm
{\it Subject classifications:} 35P15; 35C20; 35S99; 47A75.

\baselineskip = 18pt
\vfill\eject
\section{Introduction}

The study of shape resonances is a rather old subject in semiclassical analysis, and since the years 80's many mathematical works has been done in order to both locate them and estimate their widths (see, e.g., \cite{AsHa, CDKS, HeSj2, HiSi, FLM} and references therein). In particular, one should mention the work \cite{CDKS}, where the existence of shape resonances exponentially close to the real axis is proved, and the work \cite{HeSj2}, where a more refined analysis leads to optimal estimates on the widths of resonances that are near a local minimum of the potential.  For more excited shape resonances, however, only lower bounds on their widths are available in general, except for the one-dimensional case where the exact asymptotic behavior can be determined : see \cite{Se}.

As it is well known, the physical interest of such studies relies on the fact that  the widths of the resonances rare directly related to the life-time of metastable quantum states.

The purpose of this work is to extend some of the results of \cite{Se} to the multidimensional case.

More precisely, considering the semiclassical Schr\"odinger operator $P := -h^2\Delta +V(x)$ on $L^2(\R^n)$ with $n\geq 1$,
we plan to produce optimal exponential estimates on the widths of highly excited shape resonances, that is, shape resonances that tend to an energy $E_0$ greater than the local minimal of the potential $V$. In contrast with \cite{Se}, here we assume that the potential well (that is, the bounded component $U$ of $\{V\leq E_0\}$)  is connected, excluding the situation of possible interacting wells. In this situation, the general multidimensional result says that any resonance $\rho =\rho (h)$ that tends to $E_0$ as $h\to 0_+$ is such that, for any $\varepsilon >0$, one has,
$$
|\Im \rho | \leq \O (e^{-(2S_0-\varepsilon)/h})
$$
uniformly as $h\to 0_+$. Here, $S_0>0$ is the Agmon distance (that is, the degenerate distance associated with the pseud-metric $\max (V-E_0, 0)dx^2$) between $U$ and the unbounded component ${\mathcal M}$ of $\{V\leq E_0\}$.

In other words $\rho$ satisfies,
\be
\label{upperbn}
\limsup_{h\to 0_+}\, h\ln|\Im\rho| \leq -2S_0.
\ee
When $n=1$, this result is improved into (see \cite{Se}, Theorem 0.2),
\be
\label{equivn1}
\lim_{h\to 0_+}\, h\ln|\Im\rho| = -2S_0.
\ee
Here we plan to extend this improvement to the multidimensional case. Because of (\ref{upperbn}), all we need to prove is that, for any $\varepsilon >0$, there exists a constant $C_\varepsilon >0$ such that,
\be
|\Im \rho | \geq \frac1{C_\varepsilon}e^{-(2S_0+\varepsilon)/h},
\ee
for all $h>0$ small enough.

In order to produce such a  good lower bound, when $n\geq 2$ it is necessary to add an assumption on the size of the resonant state inside $U$. This assumption is actually implied by a geometric condition on the classical Hamilton flow above $U$ (see Remark \ref{geomassump}) that is automatically satisfied in the one-dimensional case. Roughly speaking, this condition says that the energy shell $\Sigma_{E_0}:=\{(x,\xi)\in \R^{2n}\, ;\, \xi^2+V(x)=E_0\}$ is sufficiently well covered by the classical Hamilton flow, in the sense that any open set intersecting $\Sigma_{E_0}$ is flowed over a whole neighborhood of $\Sigma_{E_0}$ (this can be understood as a kind of ergodicity of the flow on $\Sigma_{E_0}$).

From a technical point of view, this problem is very close to that of estimating the tunneling for a symmetric double-wells at high excited energies, as considered e.g. in \cite{Ma1} (and indeed, part of our argument will use the results of \cite{Ma1}). However, an additional difficulty comes from the fact that here, the quantity we have to study mainly involves the size of the resonant state in  ${\mathcal M}$. This means that our work will essentially consist in connecting the size of the resonant state in ${\mathcal M}$ to that in $U$, through the barrier ${\mathcal B}:=\{ V >E_0\}$. The results of \cite{Ma1} permits us to connect the size of the state in  ${\mathcal B}$ to that in $U$ only, but not its size in ${\mathcal M}$ to that in ${\mathcal B}$. Indeed, it appears that the argument of \cite{Ma1} (which is typically an argument of propagation of microlocal analyticity) does not seem easy to adapt for this last step. However, following an idea already present in \cite{DGM}, one can develop some explicit Carleman-type inequalities that permits us to cross the border between  ${\mathcal M}$ and ${\mathcal B}$, and to conclude.

\section{Notations and assumptions}

We study the spectral properties near energy 0 of the semiclassical Schr\"odinger operator,
$$
P := -h^2\Delta +V(x) \,
$$
on $L^2(\R^n)$,
where $x=(x_1,\dots ,x_n)$ is the current variable in $\R^n$ ($n\geq 1$),
$h>0$ denotes the semiclassical parameter, and $V$ represents the potential energy.

We assume,

{\bf Assumption 1.} { \it The potential $V$ is smooth and
bounded on $\R^n$,   and it
satisfies,
\begin{itemize}
\item
$\{V\leq 0\}=U\cup {\mathcal M}$ where $U$ is compact and connected, and $U\cap{\mathcal M}=\emptyset$;
\item $V $ has a strictly negative limit $ -L$ as $\vert x\vert \rightarrow \infty$.
\end{itemize}
}

This typically describes the situation where so-called shape resonances appear. In order to be able to define such resonances, we also assume,

{\bf Assumption 2.}\label{Ass2} {\it The potential $V$ extends to a bounded holomorphic functions near a complex sector of the form,
${\mathcal S}_{\delta} :=\{x\in \C^n\, ;\,  \vert \Im x\vert \leq \delta |\Re x| \}$, with $\delta >0$. Moreover  $V$ tends to its limit at $\infty$ in this sector.}

We also assume,

{\bf Assumption 3.} {\it
$E=0$  is a non-trapping energy for $V$ above ${\mathcal M}$.}

The fact that 0 is a non-trapping energy for $V$ above ${\mathcal M}$ means that, for any $(x,\xi)\in p^{-1}(0)$ with $x\in{\mathcal M}$, one has $|\exp tH_{p}(x,\xi)|\rightarrow +\infty$ as $t\rightarrow \infty$, where $p(x,\xi):=\xi^2+V(x)$ is the symbol of $P$, and $H_{p}:=(\nabla_\xi p, -\nabla_x p)$ is the Hamilton field of $p$. It is equivalent to the existence of a function $F\in C^\infty (\R^{2n};\R)$, supported near $\{ p=0\}\cap \{ x\in{\mathcal M}\}$, that satisfies,
\be
\label{fuite}
H_{p}F(x,\xi)>0 \,  \mbox{ on } \{p=0\}.
\ee
It also implies that ${\mathcal M}$
has a smooth boundary on which $\nabla V\not=0$.

For our purpose, we will also need an additional geometric assumption (that we believe to be generic).

We denote by $d_V$ the so-called Agmon distance associated with $V$, that is, the pseudo-distance associated with the pseudo-metric $\max (0, V)dx^2$. We also denote by $G$ the set of all minimal geodesics (relatively to the Agmon distance $d_V$) between $U$ and ${\mathcal M}$  that meet each boundary $\partial U$ and $\partial {\mathcal M}$ at one point only. We assume,

{\bf Assumption 4.} {\it $G$ is a finite set.}

Note that this assumption is probably purely technical only, and can hopefully be removed by refining our construction and by using a convenient partition of unity.

In the rest of the paper, we set,
$$
S_0:= d_V(U, {\mathcal M}).
$$
Thanks to our assumptions, we necessarily have $S_0>0$.

\section{Resonances}
\label{secRes}

In the previous situation, the essential spectrum of $P$ is $[-L, +\infty)$.
The resonances of $P$ can be defined by using a complex distortion in the following way (see, e.g., \cite{Hu}): Let $f \in C^\infty (\R^n, \R^n)$ such that $f(x) = x$ for $\vert x\vert$ large enough. For $\theta\not=0$ small enough, we define the distorted operator $P_{\theta}$ as the value at $ \nu = i \theta$ of  the extension to the complex of the  operator  $U_\nu P U_\nu^{-1}$ which is  defined for $\nu$ real, and analytic in $\nu$ for $\nu$ small enough, where we have set,
\be
\label{distorsion}
U_\nu \phi(x) := \det ( 1 + \nu df(x))^{1/2} \phi ( x + \nu f(x)).
\ee
By using the Weyl Perturbation Theorem, one can also see  that the essential spectrum of $P_\theta$ is given by,
$$
\sigma_{ess}(P_\theta )= e^{-2i\theta}\R_+.
$$
It is also well known that, when $\theta$ is positive, the discrete spectrum of $P_\theta$ satisfies,
$$
\sigma_{disc}(P_\theta)\subset \{ \Im z\leq 0\}.
$$
Then, those eigenvalues of $P_\theta$ that are located in the complex sector $\{ \Re z>0\, ;\,-2\theta<  {\rm arg}z \leq 0\}$
are called the resonances of $P$ \cite{Hu, HeSj2, HeMa}, they form a set denoted by ${\rm Res} (P)$  (on the contrary, when $\theta<0$, the eigenvalues of $P_\theta$ are just the complex conjugates of the resonances of $P$, and are called anti-resonances).

If $\rho$ is a resonance, the quantity $|\Im\rho|$ is called the width of $\rho$, and its physical importance comes form the fact that its inverse $|\Im\rho|^{-1}$ corresponds to the life-time of the corresponding resonant state.

Let us observe that the resonances of $P$ can also be viewed as the poles of the meromorphic extension, from $\{\Im z >0\}$, of some matrix elements of the resolvent $R(z):=(P-z)^{-1}$ (see, e.g., \cite{ReSi, HeMa}).

It is proved in \cite{HeSj1, HeSj2} that, in this situation, the resonances of $P$ near 0 are close to the eigenvalues of the operator
\be
\label{merbouche}
\widetilde P:= -h^2\Delta + \widetilde V
\ee
where $ \widetilde V\in C^\infty (\R^n; \R)$ coincides with $V$ in $\{ {\rm dis}(x,M)\geq \delta\}$ ($\delta >0$ is fixed arbitrarily small), and is such that $\inf_{\{ {\rm dis}(x,M)\leq \delta\}} \widetilde V >0$. The precise statement is the following one : Let $I(h)$ be a closed interval  containing 0, and $a(h)>0$ such that $a(h)\to 0$ as $h\to 0_+$, and, for all $\varepsilon >0$ there exists $C_\varepsilon >0$ satisfying,
\be
\label{anonexp}
a(h)\geq \frac1{C_\varepsilon} e^{-\varepsilon /h};
\ee
\be
\label{gap1}
\sigma (\widetilde P)\cap \left( (I(h)+[-2a(h), 2a(h)])\backslash I(h)\right) =\emptyset,
\ee
for all $h>0$ small enough. Then, there exists a constant $\varepsilon_1>0$ and a bijection,
$$
\widetilde\beta \, : \, \sigma (\widetilde P)\cap I(h)\,  \to \, {\rm Res} (P)\cap \Gamma (h),
$$
where we have set,
$$
\Gamma (h):= (I(h)+[-a(h), a(h))+i[-\varepsilon_1, 0],
$$
such that, for any $\varepsilon >0$, one has,
\be
\label{bij}
\widetilde\beta (\lambda) -\lambda ={\mathcal O}(e^{-(2S_0-\varepsilon)/h}),
\ee
uniformly as $h\to 0_+$.

In particular, since the eigenvalues of $\widetilde P$ are real, one obtains that, for any $\varepsilon >0$, the resonances $\rho$ in $\Gamma (h)$ satisfy,
\be
\label{upperbd}
\Im \rho ={\mathcal O}(e^{-(2S_0-\varepsilon)/h}).
\ee

From now on, we consider the particular case where $I(h)$ consists of a unique value $E(h)$, such that,
\be
\label{gap2}
\begin{aligned}
& E(h)\in \sigma_{disc} (\widetilde P);\\
& E(h)\to 0 \mbox{ as } h\to 0_+;\\
& \sigma(\widetilde P)\cap [E(h)-2a(h), E(h)+2a(h)] =\{ E(h)\}, \\
& \mbox{ where } a(h) \mbox{ satisfies (\ref{anonexp})}.
\end{aligned}
\ee

We denote by $u_0$ the normalised eigenstate of $\widetilde P$ associated with $E(h)$, and,
applying (\ref{bij}), we also denote by $\rho=\rho (h)$ the unique resonance of $P$ such that $\rho - E(h)={\mathcal O}(e^{-(2S_0-\varepsilon)/h})$.

The purpose of this paper is  to obtain a lower bound on the width $|\Im\rho|$, possibly of the same order of magnitude as the upper bound.

\section{Main Result}
\label{secmainth}

Following the ideas of \cite{Ma1}, we consider the following additional assumption of non degeneracy,

{\bf Assumption [ND]} For all $\varepsilon >0$ and for all neighborhood $W$ of the set $\bigcup_{\gamma\in G}(\gamma \cap \partial U)$, there exists $C=C(\varepsilon, W)>0$ such that, for all $h>0$ small enough, one has,
$$
\| u_0\|_{L^2(W)}\geq \frac1{C}e^{-\varepsilon /h}.
$$
Our main result is,
\begin{theorem}\sl
\label{mainth}
Suppose Assumptions 1-3, (\ref{gap2}), and Assumption [ND]. Then, for any $\varepsilon >0$, there exists $C(\varepsilon)>0$ such that, for all $h>0$ small enough, one has,
\be
\label{lowerbd}
|\Im \rho (h)| \geq \frac1{C(\varepsilon)}e^{-(2S_0+\varepsilon)/h}.
\ee
 \end{theorem}

 \begin{remark}\sl
 In view of (\ref{upperbd}), this lower bound is optimal. Indeed, a consequence of (\ref{upperbd}) and (\ref{lowerbd}) is the following identity:
 \be
 \lim_{h\to 0_+} h\ln |\Im \rho| = -2S_0
 \ee
 \end{remark}
 \begin{remark}\sl
 As in \cite{Ma1}, it can be shown that Assumption [ND] is indeed necessary to have (\ref{lowerbd}), and examples where Assumption [ND] is not satisfied can be constructed.
  \end{remark}
   \begin{remark}\sl
   \label{geomassump}
Assumption [ND] is always satisfied in the one dimensional case. When $n\geq 2$, thanks to standard properties of propagation of the microsupport (see, e.g., \cite{Ma2}), a sufficient condition to have Assumption [ND] is the following geometrical one (see also \cite{Ma1}): For any neighborhood $W$ in $\R^{2n}$ of $\bigcup_{\gamma\in G}(\gamma \cap \partial U)\times \{0\}$, the set $\bigcup_{t\in\R} \exp tH_p(W)$ is a neighborhood of $\Sigma_0:=\{\xi^2 +V(x)=0\, ,\, x\in U\}$. Obviously, an equivalent simpler formulation is: For any open set $W$ intersecting $\Sigma_0$,  $\bigcup_{t\in\R} \exp tH_p(W)$ is a neighborhood of $\Sigma_0$.
 \end{remark}

 \section{Reduction to an estimate in ${\mathcal M}$}
 From now on, we denote by $u$ the resonant state of $P$ associated with the resonance $\rho$, and normalised in such a way that,
 \be
 \| u\|_{L^2(\ddot O)} =1,
 \ee
 where $\ddot O:= \R^n\backslash {\mathcal M}$. Then, it is well known (see, e.g., \cite{HeSj2}) that for any bounded set ${\mathcal B}\subset {\mathcal M}$, and for any $\varepsilon >0$, one has,
 \be
 \label{estumer1}
  \| u\|_{L^2({\mathcal B})} ={\mathcal O}(e^{-(S_0-\varepsilon)/h}).
 \ee
 In addition, if we set,
 $$
 {\mathcal T}_1:= \bigcup_{\gamma\in G}(\gamma\cap \partial{\mathcal M})
 $$
 (the so-called set of ``points of type 1'' in the terminology of \cite{HeSj2}), and
 if ${\mathcal B}$ stays away from the set,
 $$
{\mathcal A}:=  \Pi_x \left(\bigcup_{t\in\R}\exp tH_p\left( {\mathcal T}_1\times \{0\}\right)\right),
$$
(where $\Pi_x$ stands for the natural projection $(x,\xi)\mapsto x$, and $H_p:=(\partial_\xi p, -\partial_xp)$ is the Hamilton field of $p(x,\xi ):=\xi^2 +V(x)$), then there exists $\varepsilon_0>0$ and a neighborhood ${\mathcal B}'$ of ${\mathcal B}$ such that,
 \be
  \label{estumer2}
  \| u\|_{L^2({\mathcal B}')} ={\mathcal O}(e^{-(S_0+\varepsilon_0)/h}).
 \ee
On the other hand, performing Stokes formula on an open domain $\Omega\supset \ddot O$, we see as in \cite{HeSj2}, Formula (10.65), that one has,
 \be
\label{HeSj1065} (\Im \rho )\| u\|_{L^2(\Omega)}^2 = -h^2\Im \int_{\partial\Omega} \frac{\partial u}{\partial\nu}\overline{u} ds,
 \ee
where $ds$ is the surface measure on $\partial \Omega$, and $\nu$ stands for the outward pointing unit normal to $\Omega$. Using (\ref{estumer1})-(\ref{HeSj1065}), we deduce that, for some $\varepsilon_0' >0$, one has,
\be
\label{green}
\Im \rho =-h^2\Im \int_{\partial\Omega\cap {\mathcal A}'} \frac{\partial u}{\partial\nu}\overline{u} ds+{\mathcal O}(e^{-(2S_0 +\varepsilon_0')/h}),
\ee
where ${\mathcal A}'$ is an arbitrarily small neighborhood of ${\mathcal A}$.

In order to transform this expression into a more practical one, we plan to use the analytic pseudofifferential calculus introduced in \cite{Sj}. For this purpose, we first have to prove some a priori estimate on $u$ near ${\mathcal A}$.

So, let $z_1\in {\mathcal T}_1$, let $W_1$ be a neighborhood of $z_1$ in $\partial{\mathcal M}$, and for $t_0>0$ sufficiently small, consider the two Lagrangian manifolds,
$$
\Lambda_\pm := \bigcup_{0<\pm t<2t_0} \exp t H_p (W_1\times \{0\}) \,\,\quad  (\subset \{ p=0\}).
$$
(Note that they are Lagrangian because $W_1\times\{0\}$ is isotropic.)
Then, it is easy to check that $\Lambda_\pm$ projects bijectively on the base, and since
 $p(x,\xi)$ is an even function of $\xi$,  we see that they can be represented by an equation of the type,
 $$
 \Lambda_\pm \, :\, \xi =\pm \nabla \psi (x),
 $$
 where $\psi$ is a real-analytic function,  such that,
 \be
 \label{eikon}
 (\nabla\psi (x))^2 + V(x) =0.
\ee
Now, we set $z_0:=\Pi_x  \exp t_0 H_p(z_1,0)$,
and we still denote by $\psi$ an holomorphic extension of $\psi$ to a complex neighborhood of $z_0$. We have,
\begin{proposition}\sl
For any $\varepsilon_1 >0$, one has,
$$
e^{-i\psi/h}u \in H_{-S_0+\varepsilon_1|\Im x|, z_0},
$$
where $H_{-S_0+\varepsilon_1|\Im x|, z_0}$ is the Sj\"ostrand's space consisting of $h$-dependent holomorphic functions $v=v(x;h)$ defined near $z_0$, such that, for all $\varepsilon >0$,
$$
v(x,h)={\mathcal O}(e^{(-S_0+\varepsilon_1|\Im x|+\varepsilon)/h}),
$$
uniformly for $x$ close to $z_0$ and $h>0$ small enough.
\end{proposition}
\begin{proof}
Set
\be
v(x,h):= e^{-i\psi/h+S_0/h}u(x,h).
\ee
We have to prove that $v\in H_{\varepsilon_1|Im x|, z_0}$ for all $\varepsilon_1 >0$.

Let $\chi \in \mathcal{C}_{0}^{\infty}(\R^n)$ supported in a small neighborhood of $z_0$, and such that $\chi =1$ near $z_0$. Setting $\varphi(x,y,\tau):=(x-y)\tau+\frac{1}{2}i(x-y)^2$ and 
$a(x,\tau):=1+\frac{1}{2}ix\tau$,
we can write (see, e.g., \cite{Sj}),
\be
\label{repv}
v(x)=\int e^{i|\xi|\varphi(x,y,\frac{\xi}{|\xi|})}a(x-y,\frac{\xi}{|\xi|})v(y)\chi(y)dyd\xi.
\ee
Moreover, by the results of \cite{HeSj2}, we already know that, on the real, $v$ cannot be exponentially large, that is, for any $\varepsilon >0$, one has,
$$
v=\O(e^{\varepsilon /h}) \mbox{ locally uniformly on } \R^n.
$$
In addition, $v$ is solution to
\be
\label{Qv}
((hD_x+\nabla\psi)^2 +V-\rho)v=0.
\ee
We set,
$$
Q(y, hD_y):=(hD_y+\nabla\psi)^2 +V(y)-\rho,
$$
and, in order to estimate the integral, as in \cite{Ma1}, we first plan to construct a symbol $b=b(x,y,\tau, \xi ,h)\sim \sum_{k\geq 0} b_{k}(x,y,\tau,h)|\xi|^{-k}$, with large parameter $|\xi|$, in such a way that one has,
\be
\label{asymconst}
e^{-i|\xi|\varphi(x,y,\tau)} {}^{t}{Q}(y, hD_y)\left(e^{i|\xi|\psi(x,y,\tau)}b\right)  \sim a(x-y,\tau).
\ee
Here, the asymptotic must hold as $|\xi|\to \infty$, and the quantities $\tau\in S^{n-1}$, $\mu:=\frac1{h|\xi|}\in (0,\frac1{C}]$ (with $C>0$ large enough) have to be considered as extra parameters. In particular, (\ref{asymconst}) can be rewritten as,
\be
\left[(-D_y+\mu|\xi|\nabla\psi(y)-|\xi|\nabla_y\varphi)^2 +\mu^2|\xi|^2(V-\rho)\right]b \sim a (x-y,\tau)\mu^2|\xi|^2.
\ee

Since $(\nabla\psi)^2 = E_0-V$ and $\rho \to E_0$ as $h\to 0_+$, the (leading order) coefficient $c_2$ of $|\xi|^2$ satisfies,
$$
\begin{aligned}
c_2 & =\left(\mu\nabla\psi-\nabla_y\varphi\right)^2+\mu^2(V-\rho)\\
&= (\nabla_y\varphi)^2-2\mu\nabla_y\varphi\nabla\psi +o(1)\\
& =(-\tau-i(x-y))^2+2\mu(\tau-i(x-y))\nabla\psi+o(1)\\
&= 1+\O(|x-y|+\mu)+o(1).
\end{aligned}
$$

In particular, we can solve the transport equations for $x,y$ close enough to $z_0$, and for $\mu$ small enough, that is, $|\xi |\geq C/h$ with $C>0$ sufficiently large. By the microlocal analytic theory of \cite{Sj}, we also know that the resulting formal symbol admits analytic estimates, and can therefore be re-summed into a symbol $b(x,y,\tau, \xi ,h)$ such that, for some constant $\delta >0$, one has,

\be
\label{resumsymb}
e^{-i|\xi|\varphi(x,y,\tau)} {}^{t}{Q}(y, hD_y)\left(e^{i|\xi|\psi(x,y,\tau)}b\right)  - a(x-y,\tau) =\O(e^{-\delta |\xi|}),
\ee
uniformly with respect to $\tau \in S^{n-1}$, $|\xi |\geq C/h$, $h>0$ small enough, and $x,y\in \C^n$ close enough to $z_0$. 

Then, splitting the integral in (\ref{repv}), we write,
\be
v(x)=\int_{\{|\xi|\geq \frac{C}{h}\}} + \int_{\{|\xi|\leq\frac{\varepsilon_1}{h}\}} + \int_{\{\frac{\varepsilon_1}{h}\leq|\xi|\leq\frac{C}{h}\}}.
\ee

The first integral can be estimated by using (\ref{resumsymb}), an integration by part, and the fact that $v$ solves $Qv=0$. One finds that it is $\mathcal{O}(e^{-\delta_1/h})$ for some $\delta_1>0$.

The second integral can be estimated as in \cite{Ma1}, and it is $\mathcal{O}(e^{\varepsilon_1|Im x|/h})$.

For the third integral, we make the change of variable $\xi =\eta /h$, and we find,
\be
\label{lastint}
h^{-n}\int_{\{\varepsilon_1\leq |\eta|\leq C\}} e^{i(x-y)\eta/h-|\eta|(x-y)^2/2h}a(x-y,\frac{\eta}{|\eta|})\chi(y)v(y)dyd\eta.
\ee

But, from the theory of \cite{HeSj2}, we know that if $u$ is outgoing, then, near $z_0$, the microsupport of $u$ satisfies,
$$
MS(ue^{S_0/h}) \subset \Lambda_{+}.
$$
Since $\Lambda_+=\{ \nabla\psi (y) ; y \mbox{ close to } z_0\}$, by standard rules on the microsupport we deduce,
$$
MS(v) \subset \{\eta=0\}.
$$
As a consequence,  the integral appearing in (\ref{lastint}) is $\mathcal{O}(e^{-\delta_2/h})$ for some $\delta_2>0$, and the result follows.
\end{proof}

Thanks to this proposition, we can enter the framework of the analytic pseudodifferential calculus of \cite{Sj}. We set $v:= e^{-i\psi /h}u$, and, in a complex neighborhood of $z_0$, we can write $Pu=e^{-i\psi /h}Pe^{i\psi /h}v$ as,
\be
\label{anpseudo}
Pu(x)=\frac1{(2\pi h)^n}\int_{\Gamma(x)}e^{i(x-y)\alpha_\xi /h -[(x-\alpha_x)^2+(y-\alpha_x)^2]/2h}p_\psi (\alpha_x,\alpha_\xi) v(y)dyd\alpha,
\ee
where $p_\psi $ is the symbol of $P_\psi:= e^{-i\psi /h}Pe^{i\psi /h}$, and satisfies,
\be
\label{symbppsi}
p_\psi (\alpha) =(\alpha_\xi +\nabla\psi (\alpha_x))^2 + V(\alpha_x) +{\mathcal O}(h),
\ee
and where $\Gamma(x)$ is the (singular) complex contour of integration given by,
$$
\Gamma (x)\, :\, \left\{
\begin{aligned}
& \alpha_\xi = 2i\varepsilon_1\frac{\overline{x-y}}{|x-y|}\, ;\\
& |x-y| \leq r,\,\, y\in \C^n\,\, (r \mbox{ small enough with  respect to } \varepsilon_1)\, ;\\
& |x-\alpha_x| \leq r,\, \alpha_x\in\R^n.
\end{aligned}
\right.
$$
Let us observe that the identity (\ref{anpseudo}) takes place in $H_{-S_0+\varepsilon_1|\Im x|, z_0}$, that is, modulo error terms that are exponentially smaller than $e^{(-S_0+\varepsilon_1|\Im x|)/h}$ in a complex neighborhood of $z_0$.

Taking local coordinates $(x',x_n)\in \R^{n-1}\times \R$ near $z_1$, in such a way that $dV(z_1)\cdot x = -cx_n$ with $c>0$ (and thus $T_{z_1}\partial{\mathcal M} =\{ x_n=0\}$), we see that $\nabla\psi (x)$ remains close to $(0, \sqrt x_n)$. In particular, still working in these coordinates, the symbol $-V(x)-(\xi'+\nabla_{x'}\psi (x))^2$
is elliptic along $\Gamma(x)$ (at least if $\varepsilon_1$ has been chosen sufficiently small), and with positive real part. Thus, in view of (\ref{symbppsi}), so is $(\xi_n+\partial_{x_n}\psi (x) )^2-p_\psi (x,\xi) +\rho$. As a consequence, applying the symbolic calculus of \cite{Sj}, we conclude to the existence of a pseudodifferential operator $A=A(x,hD_{x'})$, with principal symbol,
$$
a(x,\xi') =\sqrt{\rho-V(x)-(\xi'+\nabla_{x'}\psi (x))^2},
$$
such that $P_\psi -\rho$ can be factorised as,
$$
P_\psi -\rho =(hD_{x_n}+\partial_{x_n}\psi (x) +A)\circ (hD_{x_n}+\partial_{x_n}\psi (x) -A),
$$
when acting on $H_{-S_0+\varepsilon_1|\Im x|, z_0}$. Since $(P_\psi -\rho)v=0$, and $hD_{x_n}+\partial_{x_n}\psi (x) +A$ is elliptic along $\Gamma (x)$, we deduce,
\be
\label{reduction}
(hD_{x_n}+\partial_{x_n}\psi (x) -A)v=0\quad \mbox{ in } H_{-S_0+\varepsilon_1|\Im x|, z_0}.
\ee
Now, going back to (\ref{green}), and choosing $\Omega$ in such a way that its boundary contains $z_0$ and is of the form $\{x_n =\delta_0\}$  (with $\delta_0 >0$ constant) near $z_0$, the corresponding part of the integral can be written as,
$$
I_0:=-h^2\Im \int_{\{x_n=\delta_0\}\cap W_0}\left( \frac{\partial v}{\partial x_n} + \frac{i}{h} \frac{\partial \psi}{\partial x_n}v\right )\overline{v} dx'
$$
where $W_0$ is a small real neighborhood of $z_0$. Thus, using (\ref{reduction}), we obtain,
$$
I_0=-h\Re \int_{\{x_n=\delta_0\}\cap W_0} (Av)\overline{v} dx'+{\mathcal O}(e^{-(2S_0+\varepsilon_0)/h}),
$$
with $\varepsilon_0>0$. Finally, observing that the principal symbol of $A$ is strictly positive in $(z_0,0)$, and proceeding as in \cite{Ma1}, Section 2 (in particular, considering the realisation on the real of $A$), we can construct an elliptic pseudodifferential operator $B$ of order 0, such that,
$$
A =B^*B + {\mathcal O}(e^{-(S_0 + \varepsilon)/h})
$$
on $L^2(\{x_n=0\}\cap W_0)$ (with some $\varepsilon >0$). Finally, taking advantage of the ellipticity of $B$, we conclude, as in \cite{Ma1}, Lemma 2.3, that we have,
$$
I_0 \leq -\frac{h}{C}\| v\|^2_{L^2(\{x_n=0\}\cap W_0)}+ {\mathcal O}(e^{-(2S_0+\varepsilon_0)/h}),
$$
where $C, \varepsilon_0$ are positive constants. Summing up all the contributions, and observing that, in the previous formula, $v$ may be replaced by $u$ (since $\psi$ is real on the real and $v=e^{-i\psi /h}u$), we have proved,
\begin{proposition}\sl
There exist $C,\varepsilon_0>0$ such that,
$$
|\Im\rho | \geq \frac{h}{C}\| u\|_{L^2(\partial\Omega)}^2 -Ce^{-(2S_0+\varepsilon_0)/h},
$$
uniformly for $h>0$ small enough.
\end{proposition}

From now on, we proceed by contradiction : We assume the existence of $\varepsilon_1 >0$ such that,
$$
|\Im\rho | ={\mathcal O}(e^{-2(S_0+\varepsilon_1)/h}),
$$
uniformly as $h\to 0_+$ (possibly along a sequence of numbers only). By the previous proposition, this implies,
\be
\label{estuabs}
\| u\|_{L^2(\partial\Omega)}={\mathcal O}(e^{-(S_0+\varepsilon_1)/h}).
\ee

\section{Propagation across $\partial{\mathcal M}$}

The purpose of this section is to propagate the estimate (\ref{estuabs}) up to the boundary of ${\mathcal M}$ and beyond. As in the previous section, we start by working locally near some point $z_1\in {\mathcal T}_1$, and we observe that (\ref{estuabs}) actually implies that $e^{S_0/h}u$ is exponentially small near any point $z_1(t):=\Pi_x\exp tH_p (z_1,0)$, with $|t|>0$ small enough (this can be seen either by standard propagation, or more simply by the fact that the choice of $\Omega$ can be changed without altering (\ref{estuabs}), as long as its boundary stays in ${\mathcal M}$).

Near $z_1$, $u$ does not satisfy anymore sufficiently good a priori estimates that would permit us to use standard propagation (we only have $u={\mathcal O}(e^{(-\min(S_0,d_V(U,x))+\varepsilon)/h})$ for all $\varepsilon >0$, and $d_V(U,x)$ takes values strictly less than $S_0$ near $x_1$). However, we will profit from the fact that the difference between $S_0$ and $d_V(U,x)$ is controlled by $\delta (x)^{3/2}$, where $\delta (x)$ is the usual distance between $x$ and the caustic set ${\mathcal C}$ where $x\mapsto d(U,x)$ becomes singular (see \cite{HeSj2}).

For this purpose, we will use explicit Carleman estimates, in a spirit similar to that of \cite{DGM} (see also \cite{KSU}).

By assumption 3, we know that ${\mathcal T}_1$ is finite. We fix once for all $z_1\in{\mathcal T}_1$, and we will first prove that $e^{S_0/h}u$ is exponentially small near $z_1$.

As in the previous section, we take  local coordinates $(x',x_n)\in \R^{n-1}\times \R$ centered at $z_1$, in such a way that $dV(z_1)\cdot x = -cx_n$ with $c>0$. We also consider a tubular neighborhood $Z_\delta$ of $z_1$ of the form,
$$
Z_\delta :=\{ -\delta \leq x_n\leq \delta_0\, ;\, |x'|\leq \delta_0\},
$$
where $\delta_0>0$ is fixed sufficiently small, and $\delta \in(0,\delta_0)$ is a small parameter that we will possibly shrink later on. We divide the boundary of $Z_\delta$ into,
$$
\partial Z_\delta =\Sigma_\delta \cup \Sigma_0\cup\Sigma',
$$
with,
$$
\begin{aligned}
&\Sigma_\delta:=\{ x_n=-\delta  ;\, |x'|\leq \delta_0\};\\
&\Sigma_0:=\{ x_n=\delta_0  ;\, |x'|\leq \delta_0\};\\
&\Sigma':=\{ -\delta \leq x_n\leq \delta_0\, ;\, |x'|= \delta_0\}.
\end{aligned}
$$

Note that, thanks to Assumption 3 and (\ref{estumer2}), we already know that, for $\delta$ sufficiently small, there exists $\varepsilon_0>0$ independent of $\delta$ such that,
\be
\label{estsihma'}
u={\mathcal O}(e^{-(S_0+\varepsilon_0)/h}) \,\, \mbox{ on } \Sigma',
\ee
and, using the equation $Pu=\rho u$, we obtain similar estimates on the derivatives of $u$, too. Moreover, by (\ref{estuabs}), we also have,
\be
\label{estsihma0}
u={\mathcal O}(e^{-(S_0+\varepsilon_0')/h}) \,\, \mbox{ on } \Sigma_0,
\ee
for some constant $\varepsilon_0'>0$ (and the same for the derivatives of $u$).

Finally, by the same techniques as in \cite{Ma1}, Section 2, we see that, near $z_1$, the distance $d_V(x,{\mathcal M})$ satisfies,
$$
d_V(x,{\mathcal M}) ={\mathcal O}(|x_n|^{3/2}),
$$
and, by the triangle inequality, we also have,
$$
d_V(U,x)\geq S_0 - d_V(x,{\mathcal M}).
$$
As a consequence, there exists a constant $c_0>0$ such that,
$$
d_V(U,x)\geq S_0 - c_0|x_n|^{3/2}
$$
and thus, for any $\delta, \varepsilon >0$ small enough, one has,
\be
\label{estsihmadelta}
u={\mathcal O}(e^{-(S_0-c_0\delta^{3/2}-\varepsilon)/h}) \,\, \mbox{ on } \Sigma_\delta,
\ee
and similarly for the derivatives of $u$.
\begin{proposition}\sl
\label{propinM}
For $\delta >0$ sufficiently small, there exists $\varepsilon_\delta >0$ such that,
$$
\| u\|_{L^2(Z_\delta)}={\mathcal O}(e^{-(S_0+\varepsilon_\delta)/h}),
$$
uniformly for $h>0$ small enough.
\end{proposition}
\begin{proof} The proof relies on some explicit Carleman-type estimates, in a way rather similar to that of \cite{DGM}. We set,
$$
v:= e^{\alpha (x_n+\delta)/h}u,
$$
where $\alpha >0$ is fixed sufficiently small in order to have $2\alpha \delta_0 < \min(\varepsilon_0, \varepsilon_0')/2$ (here $\varepsilon_0,\varepsilon_0'$ are those of (\ref{estsihma'})-(\ref{estsihma0})). The function $v$ is solution to,
$$
(-h^2\Delta + V -\rho -\alpha^2+2h\alpha\partial_{x_n})v=0,
$$
and, by (\ref{estsihma'})-(\ref{estsihmadelta}), for any $\varepsilon>0$ and for some $\varepsilon_1>0$, we have,
\be
\label{estvbord}
\begin{aligned}
& \| v\|_{H^2(\Sigma_0\cup \Sigma')} ={\mathcal O}(e^{-(S_0+\varepsilon_1)/h});\\
& \| v\|_{H^2(\Sigma_\delta)}={\mathcal O}(e^{-(S_0-c_0\delta^{3/2}-\varepsilon)/h}).
\end{aligned}
\ee
On the other hand, we have,
\be
\label{carleman1}
\begin{aligned}
0 = &\| (-h^2\Delta + V -\rho -\alpha^2+2h\alpha\partial_{x_n})v\|^2_{L^2(Z_\delta)}\\
= & \| (-h^2\Delta + V -\rho -\alpha^2)v\|^2_{L^2(Z_\delta)}+4h^2\alpha^2 \| \partial_{x_n}v\|^2_{L^2(Z_\delta)}\\
&+ 4h\alpha\,\Re \int_{Z_\delta}(-h^2\Delta + V -\rho -\alpha^2)v\, \overline{\partial_{x_n}v}dx.
\end{aligned}
\ee
We first prove,
\begin{lemma}\sl
There exists a constant $C>0$ such that, for all $\delta >0$ small enough,
\be
\label{estCarleman}
\Re  \int_{Z_\delta}(-h^2\Delta + V -\rho -\alpha^2)v\, \overline{\partial_{x_n}v}dx\geq \frac1{C}\| v\|_{L^2(Z_{\delta})}^2-Ce^{-2(S_0-c_0\delta^{3/2}-\varepsilon)/h}.
\ee
\end{lemma}
\begin{proof} We set $Q:= 4h\alpha\Re  \int_{Z_\delta}(-h^2\Delta + V -\rho -\alpha^2)v\, \overline{\partial_{x_n}v}dx$.
By integrations by part in $x_n$, we have,
$$
\begin{aligned}
  Q= & 4h\alpha\Re \int_{\Sigma_0}\left(-h^2\Delta+V-\rho-\alpha^2\right)v\cdot\overline{v}dx'\\
  & -
4h\alpha\Re\int_{\Sigma_\delta}\left(-h^2\Delta+V-\rho-\alpha^2\right)v\cdot\overline{v}dx'\\
 & -4h\alpha\Re\int_{Z_\delta}\partial_{x_n}\left(-h^2\Delta+V-\rho-\alpha^2\right)v\cdot\overline{v}dx
 \end{aligned}
 $$
 and thus, by (\ref{estvbord}),
 
 $$
\begin{aligned}
Q = & 4h\alpha\Re\left[-\int_{\Sigma_\delta}\left(-h^2\partial_{x_n}^2-h^2\Delta_{x'}\right)v\cdot\overline{v}dx'-
\int_{\Sigma_\delta}(V-\rho-\alpha^2)v\cdot\overline{v}dx'\right]\\
& - 4h\alpha \int_{Z_\delta}\partial_{x_n}(-h^2\partial_{x_n}^2-h^2\Delta_{x'}+V-\rho-\alpha^2)v\cdot\overline{v}dx +{\mathcal O}(e^{-2(S_0+\varepsilon_1)/h})\\
 =  & 4h\alpha\Re\int_{\Sigma_\delta}h^2\partial_{x_n}^2 v\cdot \overline{v}dx'+4h\alpha\Re\int_{\Sigma_\delta}h^2\Delta_{x'}v\cdot\overline{v}dx'
-4h\alpha\Re\int_{\Sigma_\delta}(V-\rho-\alpha^2)|v|^2dx'\\
& -4h\alpha\Re\int_{Z_\delta}(-h^2\partial_{x_n}^2-h^2\Delta_{x'}+V-\rho-\alpha^2)\partial_{x_n}v\cdot\overline{v}dx\\ 
& -4h\alpha\Re\int_{Z_\delta}(\partial_{x_n}V)v\cdot\overline{v}dx +{\mathcal O}(e^{-2(S_0+\varepsilon_1)/h}).
\end{aligned}
 $$
 Then, using Green's formula in the $x'$ variables, and again (\ref{estvbord}), we obtain,
 
  $$
\begin{aligned}
Q=& 4h\alpha\Re\int_{\Sigma_\delta}h^2\partial_{x_n}^2v\cdot\overline{v}dx'-
4h\alpha\Re\int_{\Sigma_\delta}h^2|\nabla_{x'}|^2dx'-4h\alpha\Re\int_{\Sigma_\delta}(V-\rho-\alpha^2)|v|^2dx'\\
& -4h\alpha\Re\int_{Z_\delta}-h^2\partial_{x_n}^3v\cdot\overline{v}dx
-4h\alpha\Re\int_{Z_\delta}-h^2\Delta_{x'}\partial_{x_n}v\cdot\overline{v}dx\\
& - 4h\alpha\Re\int_{\Sigma_\delta}(V-\rho-\alpha^2)\partial_{x_n}v\cdot\overline{v}dx\\
& - 4h\alpha\Re\int_{Z_\delta}(\partial_{x_n}V)v\cdot\overline{v}dx+{\mathcal O}(e^{-2(S_0+\varepsilon_1)/h})\\
= & 4h\alpha\Re\int_{\Sigma_\delta}h^2\partial_{x_n}^2v\cdot\overline{v}dx'-4h\alpha\int_{\Sigma_\delta}h^2|\nabla_{x'}|^2dx'\\
&-4h\alpha\Re\int_{\Sigma_\delta}(V-\rho-\alpha^2)|v|^2dx'-4h\alpha\Re\int_{Z_\delta}h^2\partial_{x_n}^2v\partial_{x_n}\overline{v}dx\\
& -4h\alpha\Re\int_{\Sigma_\delta}h^2\partial_{x_n}^2v\cdot\overline{v}dx'
-4h\alpha\Re\int_{Z_\delta}(\partial_{x_n}v)\overline{\left(-h^2\Delta_{x'}+(V-\rho-\alpha^2)\right)v}dx\\
& -4h\alpha\Re\int_{Z_\delta}(\partial_{x_n}V)v\cdot\overline{v}dx+{\mathcal O}(e^{-2(S_0+\varepsilon_1)/h}).
\end{aligned}
$$

By an integration by parts, we also see that,

$$
\Re\int_{Z_\delta}\partial_{x_n}^2v\partial_{x_n}\overline{v}dx=-\frac12 \int_{\Sigma_\delta}|\partial_{x_n}v|^2dx' +{\mathcal O}(e^{-2(S_0+\varepsilon_1)/h})
$$

and therefore,

 $$
\begin{aligned}
Q= & -4h\alpha\int_{\Sigma_\delta}h^2|\nabla_{x'}|^2dx'-4h\alpha\Re\int_{\Sigma_\delta}(V-\rho-\alpha^2)|v|^2dx'\\
&+2h\alpha \int_{\Sigma_{\delta}}h^2|\partial_{x_n}v|^2dx' -\overline{Q} -4h\alpha\Re\int_{Z_\delta}(\partial_{x_n}V)v\cdot\overline{v}dx+{\mathcal O}(e^{-2(S_0+\varepsilon_1)/h}).
\end{aligned}
$$
Using that $\overline Q =Q$, we finally obtain,
$$
\begin{aligned}
Q= &-2h^3\alpha\| \nabla_{x'}v\|^2_{\Sigma_\delta}+h^3\alpha\| \partial_{x_n}v\|^2_{\Sigma_\delta}-2h\alpha\int_{Z_{\delta}}(\partial_{x_n}V)|v|^2 dx\\
& -2h\alpha\Re\int_{\Sigma_{\delta}}(V-\rho-\alpha)|v|^2 dx'+{\mathcal O}(e^{-2(S_0+\varepsilon_1)/h})\\
=& -2h\alpha\int_{Z_{\delta}}(\partial_{x_n}V)|v|^2 dx+{\mathcal O}(e^{-2(S_0-c_0\delta^{3/2}-\varepsilon)/h}),
\end{aligned}
$$
where $\varepsilon >0$ is arbitrarily small.
Since $\partial_{x_n}V(z_1) = -c <0$, we deduce the result of the lemma.
\end{proof}

Inserting (\ref{estCarleman}) into (\ref{carleman1}), we obtain,
$$
\| v\|_{L^2(Z_{\delta})}^2={\mathcal O}(e^{-2(S_0-c_0\delta^{3/2}-\varepsilon)/h}),
$$
and therefore, since $2\alpha (x_n+\delta)\geq \alpha\delta$ on $Z_{\delta /2}\subset Z_{\delta}$,
$$
\| u\|_{L^2(Z_{\delta/2})}^2={\mathcal O}(e^{-(2S_0+\alpha\delta -2c_0\delta^{3/2}-2\varepsilon)/h}).
$$
Observing that $\alpha$ has been chosen independently of $\delta$, we obtain the result of Proposition \ref{propinM} by taking $\delta$ sufficiently small in order to have $\alpha\delta >2c_0\delta^{3/2}$.
\end{proof}

Summing up, using Lemma \ref{propinM} at each point of ${\mathcal T}_1$ and (\ref{estumer2}), we obtain,
\begin{proposition}\sl
\label{propintM}
There exists of a neighborhood ${\mathcal V}$ of $\partial{\mathcal M}$, and a constant $\varepsilon_2>0$, such that,
$$
\| u\|_{L^2({\mathcal V})} ={\mathcal O}(e^{-(S_0+\varepsilon_2)/h}),
$$
uniformly for $h>0$ small enough.
\end{proposition}

\section{Completion of the proof}

From this point, the proof proceeds exactly as in \cite{Ma1}. More precisely, if $\widetilde P=-h^2\Delta + V$ is the operator defined as in (\ref{merbouche}), with $\widetilde V =V$ near $\ddot O\backslash{\mathcal V}'$ (where ${\mathcal M}\subset\subset{\mathcal V}' \subset\subset{\mathcal V}$), we already know (see \cite{HeSj2}, Theorem 9.9) that the difference $u-u_0$ satisfies,
$$
\| u-u_0\|_{L^2(\ddot O\backslash{\mathcal V}')} ={\mathcal O}(e^{-(S_0+\varepsilon_3/h}),
$$
for some constant $\varepsilon_3>0$. Then, applying Proposition \ref{propintM}, we deduce,
$$
\| u_0\|_{L^2({\mathcal V}\backslash{\mathcal V}')}={\mathcal O}(e^{-(d_V(U,x)+\varepsilon_3)/h}),
$$
with $\varepsilon_4 >0$. At this point, we are in a situation absolutely similar to that of \cite{Ma1}. In particular, the previous estimate can be propagated up to the well $U$ along any minimal geodesic $\gamma\in G$, and as in \cite{Ma1}, Section 6, we obtain that for all $x_1\in \bigcup_{\gamma\in G}(\gamma \cap \partial U)$, one has,
$$
(x_1,0)\notin MS (u_0),
$$
where $MS(u)$ stands for the microsupport of $u$ as defined, e.g., in \cite{Ma2} (it was called $FS_a(u)$ in \cite{Ma1}). Since, in addition, $MS(u)\subset p^{-1}(0)$, and,
$$
(\partial U\times \R^n)\,\cap\, p^{-1}(0) \subset \{ \xi =0\},
$$
we deduce,
$$
MS(u)\, \cap \, \left( \bigcup_{\gamma\in G}(\gamma \cap \partial U)\times \R^n\right) =\emptyset,
$$
and thus by standard properties of $MS(u)$ (see, e.g., \cite{Ma2}), we conclude to the existence of a neighborhood $W$ of $\bigcup_{\gamma\in G}(\gamma \cap \partial U)$ and of a positive constant $\varepsilon_5>0$, such that,
$$
\| u\|_{L^2(W)}={\mathcal O}(e^{-\varepsilon_5/h}),
$$
uniformly for $h>0$ small enough. But this is in contradiction with Assumption [ND], and the proof of Theorem \ref{mainth} is complete.



{}


\begin{thebibliography}{}

\vskip 0.5cm

\bibitem[AsHa]{AsHa} Ashbaugh, M., Harrell, E.M. {\it Perturbation theory for shape resonances and large barrier potentials},  Comm. Math. Phys. 83, 151-170 (1982)

\bibitem[CDKS]{CDKS} Combes, J.M., Duclos, P., Klein, M., Seiler, R. {\it The shape resonance},  Comm. Math. Phys. 110, 215-236 (1987)


\bibitem[DGM]{DGM} Duykaerts, T., Grigis, A., Martinez, A. {\it Resonance widths for general Helmholtz Resonators with straight neck},  Preprint 2015

\bibitem[FLM] {FLM} Fujiie, S., Lahmar-Benbernou, A., Martinez, A., {\it Width of shape resonances for non globally analytic potentials}, to appear in Journal of the Mathematical Society of Japan (2011).

\bibitem[GrMa] {GrMa} Grigis, A., Martinez, A. {\it Resonance widths for the molecular predissociation},  Analysis \& PDE 7-5 (2014), 1027--1055. DOI 10.2140/apde.2014.7.1027

\bibitem[HeMa] {HeMa} Helffer, B., Martinez, A., {\it Comparaison entre les diverses notions de r\'esonances},
Helv. Phys. Acta, Vol.60 (1987),no.8, pp.992-1003.

\bibitem[HeSj1] {HeSj1} Helffer, B., Sj\"ostrand, J., {\it Multiple Wells
in the Semiclassical Limit I}, Comm. in P.D.E. {\bf 9(4)} (1984), pp.337-408.

\bibitem[HeSj2] {HeSj2} Helffer, B., Sj\"ostrand, J., {\it R\'esonances
en limite semi-classique },  Bull.
Soc. Math. France 114, Nos. 24-25 (1986).

\bibitem[HiSi]{HiSi} Hislop, P.D., Sigal, I.M. {\it Semiclassical theory of  shape resonances in quantum mechanics},  Memoirs of the AMS 399, 1-123 (1989)

\bibitem[Hu] {Hu} Hunziker,  W., {\it Distortion analyticity and molecular resonance curves}, Ann. Inst. H. Poincar\'e Phys. Th\'eor. 45 (1986), no. 4, pp. 339-358.


\bibitem[KSU]{KSU}
C. E. Kenig, J.~Sj\"ostrand, G. Uhlmann.
\newblock The Calder\'on problem with partial data.
\newblock {\em  Ann. of Math.} (2) 165 (2007), no. 2, 567-591.

\bibitem[Ma1] {Ma1} Martinez, A., {\it Estimations de l'effet tunnel pour le double puit II, \'etats hautement excit\'es}, Bull. Soc. Math. France, 116 (1988), 199-229

\bibitem[Ma2] {Ma2} Martinez, A., {\it An Introduction to Semiclassical and
Microlocal Analysis}, UTX Series, Springer-Verlag New-York (2002).

\bibitem[ReSi] {ReSi} Reed, M., Simon, B., {\it Methods of Modern Mathematical Physics IV: Analysis of Operators}, Academic Press INC., New York, 1978.

\bibitem[Se] {Se} Servat, E., {\it R\'esonances en dimension un pour l'op\'erateur de Schr\"odinger}, Asymptotic Analysis 39 (2004) 187-224

\bibitem[Si] {Si} Simon, B., {\it Semiclassical analysis if low lying eigenvalues I. Non-degenerate minima: Asymptotic expansions}, Annales Inst. H. Poincar\'e, Section A, 38, No. 3, 1983, 295-308.

\bibitem[Sj] {Sj} Sj\"ostrand J., {\it Singularit\'es analytiques microlocales}, Ast\'erisque 95 (1982)


\end{thebibliography}
\end{document}